\documentclass[a4paper]{article}
\usepackage{RR}
\usepackage{hyperref}
\usepackage{a4wide}

\usepackage{graphicx}
\usepackage{amssymb}
\usepackage{epsfig}
\usepackage{tabularx}
\usepackage{latexsym}

\newtheorem{lemma}{Lemma}[section]

\newtheorem{corollary}{Corollary}[section]
\newtheorem{theorem}{Theorem}[section]

\newtheorem{defn}[theorem]{Definition}

\def\R{\mathbb{R}}
\def\P{\mathbb{P}}
\def\X{\mathbb{X}}

\RRdate{Septembre 2007}

\RRauthor{
Mridul Aanjaneya\thanks{Department of Computer Science and 
Engineering, IIT Kharagpur, 721302, India [Email: mridul@cse.iitkgp.ernet.in] 
This problem was conceptualized and solved while the first author was invited 
under the INRIA \textit{Internship Program} for a summer visit.}
\and
Monique Teillaud\thanks{INRIA Sophia Antipolis, BP 93, 06902 Sophia Antipolis cedex (France) 
[Email: Monique.Teillaud@sophia.inria.fr - http://www-sop.inria.fr/geometrica/team/Monique.Teillaud/]}
}
\authorhead{M. Aanjaneya \& M. Teillaud}
\RRetitle{Triangulating the Real Projective Plane}
\RRtitle{Trianguler le plan projectif r\'eel}
\titlehead{Triangulating the Real Projective Plane}

\RRabstract{
We consider the problem of computing a triangulation of the real
projective plane $\mathbb{P}^2$, given a finite point set $\mathcal{P}
= \{p_1, p_2,\ldots, p_n\}$ as input. We prove that a triangulation of
$\mathbb{P}^2$ always exists if at least six points in $\mathcal{P}$
are in {\it general position}, i.e., no three of them are collinear.
We also design an algorithm for triangulating $\mathbb{P}^2$ if this
necessary condition holds. As far as we know, this is the first
computational result on the real projective plane.
}
\RRresume{Nous consid\'erons le calcul de la triangulation d'un
ensemble fini de points $\mathcal{P} = \{p_1, p_2,\ldots, p_n\}$ dans
le plan projectif. Nous d\'emontrons que la triangulation existe
toujours d\`es lors qu'au moins six points de $\mathcal{P}$ sont en
position g\'en\'erale, c'est-\`a-dire que trois d'entre eux ne sont
jamais align\'es. Nous proposons \'egalement un algorithme pour
trianguler $\mathbb{P}^2$ si cette condition n\'ecessaire est
remplie. \`A notre connaissance, c'est le premier r\'esultat
algorithmique connu pour le plan projectif r\'eel.}

\RRkeyword{Computational geometry, triangulation, simplicial complex, projective geometry, algorithm}
\RRmotcle{G\'eom\'etrie algorithmique, triangulation, complexe
simplicial, g\'eom\'etrie projective, algorithme}

\RRprojets{Geometrica}
\RRtheme{\THSym} 
\URSophia 
\RRNo{6296}

\begin{document}
\makeRR

\section{Introduction\label{intro}}

The real projective plane $\mathbb{P}^2$ is in one-to-one
correspondence with the set of lines of the vector space
$\mathbb{R}^3$. Formally, $\P^2$ is the quotient $\P^2 = \R^3 \!\! -
\!\! \{0\} \;\;/\sim$ where the equivalence relation $\sim$ is defined
as follows: for two points $p$ and $p'$ of $\P^2$, $p \sim p'$ if
$p=\lambda p'$ for some $\lambda\in\R\!\! - \!\! \{0\}$.

Triangulations of the real projective plane $\mathbb{P}^2$ have been
studied quite well in the past, though mainly from a graph-theoretic
perspective. A {\it contraction} of and edge $e$ in a map
$\mathcal{M}$ removes $e$ and identifies its two endpoints,
if the graph obtained by this operation is simple. $\mathcal{M}$ is
{\it irreducible} if none of its edges can be contracted.  Barnette
\cite{b-gtpp-82} proved that the real projective plane admits exactly
two irreducible triangulations, which are the complete graph $K_6$
with six vertices and $K_4 + \overline{K}_3$ (i.e., the
quadrangulation by $K_4$ with each face subdivided by a single
vertex), which are shown in Figure~\ref{fig-irred}. Note that these
figures are just graphs, i.e. the horizontal and vertical lines do
not imply collinearity of the points.

\begin{figure}[htbp]
\centerline{\scalebox{.6}{\includegraphics{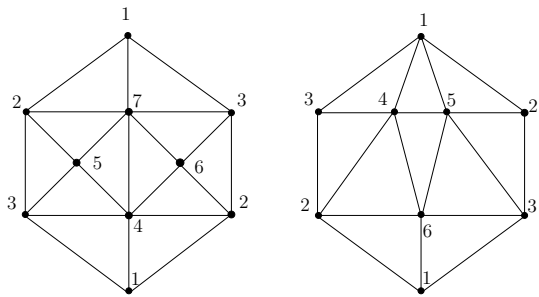}}}
\caption{The two irreducible triangulations of $\mathbb{P}^2$.\label{fig-irred}}
\end{figure}

A {\it diagonal flip} is an operation which replaces an edge $e$ in
the quadrilateral $D$ formed by two faces sharing $e$ with another
diagonal of $D$ (see Figure~\ref{fig-flip}). If the resulting graph is
not simple, then we do not apply it. Wagner \cite{w-bzv-36} proved that
any two triangulations on the plane with the same number of vertices
can be transformed into each other by a sequence of diagonal flips, up
to isotopy. This result has been extended to the torus \cite{d-wttg-73},
the real projective plane and the Klein bottle \cite{nw-dtts-90}. 
Moreover, Negami has proved that for any closed surface $F^2$, there
exists a positive integer $N(F^2)$ such that any two 
triangulations $G$ and $G'$ on $F^2$ with $|V(G)| = |V(G')| \geq
N(F^2)$ can be transformed into each other by a sequence of diagonal
flips, up to homeomorphism \cite{n-dfts-94}. Mori and Nakamoto
\cite{mn-dfhtp-05} gave a linear upper bound of $(8n-26)$ on the number
of diagonal flips needed to transform one triangulation of
$\mathbb{P}^2$ into another, up to isotopy. There are many papers
concerning with diagonal flips in triangulations, see
\cite{n-dftcs-99,e-hefg-07} for more references.

\begin{figure}[htbp]
\centerline{\scalebox{.6}{\includegraphics{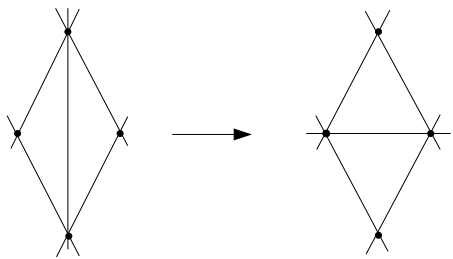}}}
\caption{A diagonal flip.\label{fig-flip}}
\end{figure} 

In this paper, we address a different problem, which consists in
computing a triangulation of the real projective plane, given a finite
point set $\mathcal{P} = \{p_1, p_2,\ldots, p_n\}$ as input. 

\begin{defn} Let us recall background definitions here. More extensive
definitions are given for instance in \cite{z-tc-05,h-cit-79}.
\begin{itemize}
\item {} An (abstract) \textit{simplicial complex} is a set $K$
together with a collection $S$ of subsets of $K$ called (abstract)
\textit{simplices} such that:
	\begin{enumerate}
	\item {} For all $v \in K$, $\{v\}\in S$. The sets $\{v\}$
are called the \textit{vertices} of $K$.
	\item {} If $\tau\subseteq\sigma \in S$, then $\tau\in S$. 
	\end{enumerate} 
Note that the property that $\sigma,\sigma' \in K \Rightarrow
\sigma\cap\sigma'\leq \sigma,\sigma'$ can be deduced from this. 
\item {} $\sigma$ is a $k$-simplex if the number of its vertices is $k+1$. If
$\tau\subset\sigma$, $\tau$ is called a face of $\sigma$. 
\item {} A \textit{triangulation} of a topological space $\X$ is a
simplicial complex $K$ such that the union of its simplices is
homeomorphic to $\X$.
\end{itemize} 
\label{def-tr}
\end{defn}

All algorithms known to compute a triangulation of a set of points in
the Euclidean plane use the orientation of the space as a fundamental
prerequisite. The projective plane is not orientable, thus none of
these known algorithm can extend to $\mathbb{P}^2$.

We will always represent $\mathbb{P}^2$ by the {\it sphere model}
where a point $p$ is same as its diametrically opposite ``copy" (as
shown in Figure~\ref{fig-sphere}(a)). We will refer to this sphere as
the {\it projective sphere}.  A {\it triangulation} of the real
projective plane $\mathbb{P}^2$ is a simplicial complex such that each
face is bounded by a 3-cycle, and each edge can be seen as a greater arc
on the projective sphere. We will also sometimes refer to a
triangulation of the projective plane as a {\it projective
triangulation}. 

\begin{figure}[htbp]
\centerline{
\scalebox{.6}{\includegraphics{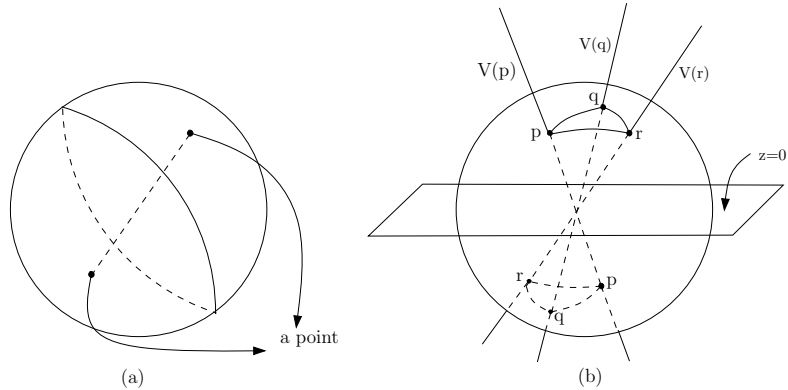}}
}
\caption{(a) The sphere model of $\mathbb{P}^2$. (b) $\bigtriangleup
pqr$ separated from its ``copy" by a {\it distinguishing plane} in
$\mathbb{R}^3$.\label{fig-sphere}} 
\end{figure}

Stolfi \cite{s-opgfg-91} had described a computational model for
geometric computations: the {\it oriented projective plane}, where a
point $p$ and its diametrically opposite ``copy" on the projective
sphere are treated as two different points. In this model, two
diametrically opposite triangles are considered as different, so, the
computed triangulations of the oriented projective plane are actually
not triangulations of $\P^2$. Identifying in practice a triangle and
its opposite in some data-structure is not straightforward.  Let us
also mention that the oriented projective model can be pretty costly
because it involves the duplication of every point, which can be a
serious bottleneck on available memory in practice.

The reader should also note that obvious approaches like triangulating
the convex hull of the points in $\mathcal{P}$ and their diametrically
opposite ``copies" (on the projective sphere) separately will not
work: it may happen that the resulting structure is not a
simplicial complex (see Figure~\ref{fig-ch} for the most obvious
example), so, it is just not a triangulation (see
definition~\ref{def-tr}).

\begin{figure}[htbp]
\centerline{\scalebox{.6}{\includegraphics{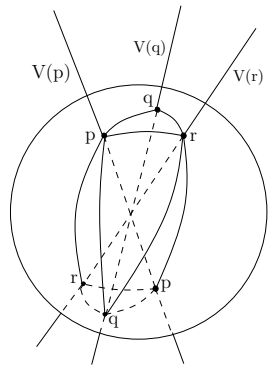}}}
\caption{The convex hull of duplicated points is not a triangulation.
\label{fig-ch}}
\end{figure}

\subsection{Terminology and Notation}

We assume that the positions of points and lines are stored as
homogeneous coordinates in the real projective plane. Positions of
points will be represented by triples $(x,y,z)$ (with $z\neq 0$) and
their coordinate vectors will be denoted by small letters like $p, q,
r,\ldots$. Positions of lines will also be represented by triples
$[x,y,z]$ but their coordinate vectors will be represented by capital
letters like $L, M, N,\ldots$. We shall also state beforehand whether
a given coordinate vector is that of a line or a point to avoid
ambiguity.  Point $p$ and line $M$ are incident if and only if the dot
product of their coordinate vectors $p\cdot M=0$. If $p$ and $q$ are
two points then the line $L=pq$ can be computed as the cross product
$p\times q$ of their coordinate vectors. Similarly the intersections
of two lines $M$ and $N$ can be computed as the cross product $M\times
N$ of their coordinate vectors.

We denote the line in $\mathbb{R}^3$ corresponding to a point $p$ in
$\mathbb{P}^2$ by $V(p)$. A plane in $\mathbb{R}^3$ which separates
$\bigtriangleup abc$ from its diametrically opposite ``duplicate copy"
on the projective sphere will be referred to as a {\it distinguishing
plane} for the given triangle (see Figure~\ref{fig-sphere}(b)).  Note
that a distinguishing plane is not unique for a given triangle. Also
note that such a plane is defined only for non-degenerate triangles on
the real projective plane.


\subsection{Contents of the Paper}

We first prove a necessary condition for the existence of a
triangulation of the set $\mathcal{P} = \{p_1, p_2,\ldots, p_n\}$ of
$\mathbb{P}^2$. More precisely, we show that such a triangulation
always exists if at least six points in $\mathcal{P}$ are in {\it
general position}, i.e., no three of them are collinear.  So if the
number of points $n$ in $\mathcal{P}$ is very large, the probability
of such a set of six points to exist is high, implying that it is
almost always possible to triangulate $\mathbb{P}^2$ from a point set.

We design an algorithm for computing a projective triangulation of
$\mathcal{P}$ if the above condition holds. The efficiency of the
algorithm is not our main concern in this paper. The existence of an
algorithm for computing a triangulation directly in $\P^2$ is our main
goal. As far as we know, this is the first computational result on the
real projective plane.

The paper is organized as follows.  In section~\ref{sec-interior}, we
devise an ``in-triangle" test for checking whether a point $p$ lies
inside a given $\bigtriangleup abc$.  In section~\ref{sec-comput}, we
first prove that a triangulation of $\mathbb{P}^2$ always exists if at
least six points in the given point set $\mathcal{P}$ are in general
position. We then describe our algorithm for triangulating
$\mathbb{P}^2$ from points in $\mathcal{P}$.  Finally, in
section~\ref{sec-conclusion}, we present some open problems and future
directions of research in this area.

\section{The Notion of ``Interior" in the Real Projective Plane\label{sec-interior}} 

It is well-known that the real projective plane is a non-orientable
surface. However, the notion of ``interior" of a closed curve exists
because the projective plane with a {\it cell} (any figure
topologically equivalent to a disk) cut out is topologically
equivalent to a M\"{o}bius band \cite{h-cit-79}.  For a given
triangle on the projective plane, we observe that its interior can be
defined unambiguously if we associate a {\it distinguishing plane}
with it.  The procedure for associating such a plane with any given
triangle will be described in Section~\ref{sec-comput}. For now we
will assume that we have been given $\bigtriangleup pqr$ along with
its distinguishing plane in $\mathbb{R}^3$. We further assume for
simplicity that this plane is $z=0$ for the given $\bigtriangleup pqr$
(as shown in Figure~\ref{fig-sphere}(b)). Consider the three lines
$V(p), V(q)$ and $V(r)$ in $\mathbb{R}^3$. These lines give rise to
four {\it double cones}, three of which are cut by the distinguishing
plane. We define the {\it interior} of $\bigtriangleup pqr$ as the
double cone in $\mathbb{R}^3$ which is not cut by its distinguishing
plane. Based on the above definition, we define a many-one mapping
$s:\mathbb{P}^2 \rightarrow \mathbb{R}^3$ from points in
$\mathbb{P}^2$ to points in $\mathbb{R}^3$ as follows:

\[
s(p) = s(x,y,z) = \left\{ \begin{array}{ll}
           (1,\frac{x}{z},\frac{y}{z}) & z\neq0; \\
           (0,1,\frac{y}{x})           & z=0, x\neq0;\\
           (0,0,1)                     & z=0, x=0. \end{array} \right. 
\]
Given three points $a=(x_0,y_0,z_0), b=(x_1,y_1,z_1), c=(x_2,y_2,z_2)$
and a point $p=(x,y,z)$, $p$ lies inside $\bigtriangleup abc$ if

\begin{eqnarray} 
sign\left| \begin{array}{ccc}
s_0 \\
s_{1} \\
s \end{array} \right| + 
sign\left| \begin{array}{ccc}
s_1 \\
s_2 \\
s \end{array} \right| + 
sign\left| \begin{array}{ccc}
s_2 \\
s_0 \\
s \end{array} \right| \mbox{$=\pm 3$}
\end{eqnarray}
and it lies on the perimeter of $\bigtriangleup abc$ if 

\begin{eqnarray}
sign\left| \begin{array}{ccc}
s_0 \\
s_{1} \\
s \end{array} \right| +
sign\left| \begin{array}{ccc}
s_1 \\
s_2 \\
s \end{array} \right| +
sign\left| \begin{array}{ccc}
s_2 \\
s_0 \\
s \end{array} \right| \mbox{$=\pm 2$}
\end{eqnarray}
Here $s_i=s(x_i,y_i,z_i)$ for $i=0,1,2$, and $s=s(x,y,z)$. The
function $sign(m)$ returns 1 if $m$ is positive, 0 if $m$ is zero, and
$-1$ if $m$ is negative.  The reader should note that similar to the
oriented projective model \cite{s-opgfg-91}, there is no notion of
interior when all $a,b,c$ and $p$ are at infinity. We now consider the
case when the distinguishing plane in $\mathbb{R}^3$ is $\alpha
x+\beta y+\gamma z=0$, where $\alpha,\beta,\gamma$ are arbitrary
constants. We use a linear transformation matrix $\mathcal{M}$ for
transforming the given plane into the plane $z=0$ according to the
equation $\mathcal{M}\cdot p' = p$, where orientation is
preserved. This transformation takes the coordinate vector $p$ of a
point to the vector $p'$. Now the $s$-mapping of equation (1) can be
used for the ``in-triangle" test with the new coordinate vectors, as
described above. For the case when $\gamma\neq0$, we have

\begin{eqnarray}
\mathcal{M} = \left[ \begin{array}{ccc}
0 & \frac{-(\beta^2+\gamma^2)}{\sqrt{(\beta^2+\gamma^2)(\alpha^2+\beta^2+\gamma^2)}} & \frac{\alpha}{\sqrt{(\alpha^2+\beta^2+\gamma^2)}} \\
\frac{\gamma}{\sqrt{\beta^2+\gamma^2}} & \frac{\alpha\beta}{\sqrt{(\beta^2+\gamma^2)(\alpha^2+\beta^2+\gamma^2)}} & \frac{\beta}{\sqrt{(\alpha^2+\beta^2+\gamma^2)}} \\
\frac{-\beta}{\sqrt{\beta^2+\gamma^2}} & \frac{\alpha\gamma}{\sqrt{(\beta^2+\gamma^2)(\alpha^2+\beta^2+\gamma^2)}} & \frac{\gamma}{\sqrt{(\alpha^2+\beta^2+\gamma^2)}} 
\end{array} \right]
\end{eqnarray}
For the case when $\gamma=0, \beta\neq0$, we have

\begin{eqnarray}
\mathcal{M} = \left[ \begin{array}{ccc}
\frac{\beta}{\sqrt{\alpha^2+\beta^2}} & 0 & \frac{\alpha}{\sqrt{\alpha^2+\beta^2}} \\
\frac{-\alpha}{\sqrt{\alpha^2+\beta^2}} & 0 & \frac{\beta}{\sqrt{\alpha^2+\beta^2}} \\
0 & -1 & 0 \end{array} \right]
\end{eqnarray}
Finally, we have the case when $\gamma=\beta=0$. In this case, we simply 
make the $X$-axis the new $Y'$-axis, the 
$Y$-axis the new $Z'$-axis, and the $Z$ axis the new $X'$-axis. 
So our tranformation matrix $\mathcal{M}$ is as follows: 

\begin{eqnarray}
\mathcal{M} = \left[ \begin{array}{ccc}
0 & 0 & 1 \\
1 & 0 & 0 \\
0 & 1 & 0 \end{array} \right]
\end{eqnarray}

Note that all the transformation matrices given by equations (2), (3) and (4) 
are {\it orthogonal} matrices, i.e., $\mathcal{M}^{-1}=\mathcal{M}^T$.

\section{Computing the Projective Triangulation of a Point Set\label{sec-comput}}

We now proceed to discuss our algorithm for triangulating the real projective 
plane given a point set $\mathcal{P} = \{p_1,p_2,...,p_n\}$ as input. 
We number the points in this section for diagramatic clarity.  
We first prove the following simple result for point sets in $\mathbb{P}^2$:

\begin{figure}[htbp]
\centerline{\scalebox{.55}{\includegraphics{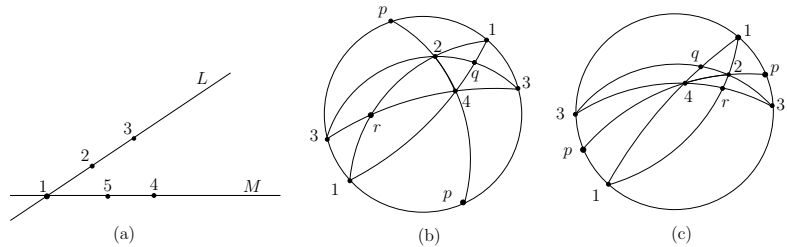}}}
\caption{(a) Every point set $\mathcal{P}$ has a {\it
$K_4$-quadrangulation}, unless $(n-1)$ points are collinear. (b) and
(c) A $K_4$-quadrangulation is sufficient for triangulating the real
projective plane.\label{fig-quad}} 
\end{figure}

\begin{lemma}
If among every set of four points in the point set $\mathcal{P}$  
at least three points are collinear, then at least $(n-1)$ points 
in $\mathcal{P}$ are collinear.  
\end{lemma}  
\begin{proof}
It is easy to see that the lemma holds if all points in $\mathcal{P}$
are collinear. So we may safely assume that this is not the case. We
prove the above lemma by the method of contradiction. Assuming that no
set of $(n-1)$ points are collinear. Consider a set of four points
$\{1, 2, 3, 4\}$ as shown in Figure~\ref{fig-quad}(a). Since among
every set of four points, at least three are collinear, so we assume
that $1, 2$ and $3$ are collinear. Now consider a fifth point $5$
instead of $3$, and assume that it does not lie on the line
$L=12$. From the given condition, it must lie on the line $M=14$. But
now no three points among the set $\{2, 3, 4, 5\}$ are collinear, a
contradiction!
\end{proof}

\begin{corollary}
\label{cor1}
If $(n-1)$ points are not collinear in the given point set
$\mathcal{P}$, then there exists a set of four points in
$\mathcal{P}$, no three of which are collinear.
\end{corollary}
We call such a set of four points a {\it $K_4$-quadrangulation}. 
Corollary \ref{cor1} states that every point
set $\mathcal{P}$ in which no $(n-1)$ points are collinear contains a
$K_4$-quadrangulation. We now make the following important observation
that such a point set can be used to construct a triangulation of the
projective plane (see Figure~\ref{fig-quad}(b,c)). We have the
following lemma:

\begin{lemma}
A $K_4$-quadrangulation can be used to construct a projective triangulation.
\end{lemma}
\begin{proof}
Consider the points of the $K_4$-quadrangulation on the projective
sphere and construct the lines (great circles) $\{12, 13, 14, 24, 23,
34\}$ (see Figure~\ref{fig-quad}(b,c)). The intersection of these
six lines define three more points $\{p, q, r\}$. We call these points
{\it pseudo-points} because these may or may not be points in
$\mathcal{P}$. It is now easy to see that the resulting triangulation
is a simplicial complex and is isomorphic to the projective
triangulation shown in Figure~\ref{fig-irred}(a).
\end{proof}
 
The reader should observe that every triangle in the above
triangulation has precisely two copies on the projective sphere which
are diametrically opposite (see Figure~\ref{fig-quad}(b,c)).  So it
now becomes possible to associate a distinguishing plane with each
triangle in the above triangulation unambiguously.  For every
$\bigtriangleup abc$ in the projective triangulation, we can take the
plane through the center $O$ of the projective sphere and parallel to
the plane passing through the end-points $a,b,c$ of one copy of
$\bigtriangleup abc$. Given a query point $u$, we can now determine
the triangle inside which it lies. We will use this fact quite
extensively in our algorithm. The procedure described above is
incomplete in the sense that we triangulate the real projective plane
with the help of some pseudo-points. We now give a necessary condition
for computing a projective triangulation from a point set
$\mathcal{P}$. The reader should observe that every triangle in
Figure~\ref{fig-quad}(b,c) is incident with exactly one
pseudo-point. We will refer to the set of triangles incident to one
pseudo-point as a {\it region}. Note that any two regions have the
same set of vertices. For constructing a projective triangulation from
$\mathcal{P}$ we will initially take help of pseudo-points, but we
will go on deleting them as their use is over. We now present the
following lemma:

\begin{figure}[htbp]
\begin{center}
\scalebox{.6}{\includegraphics{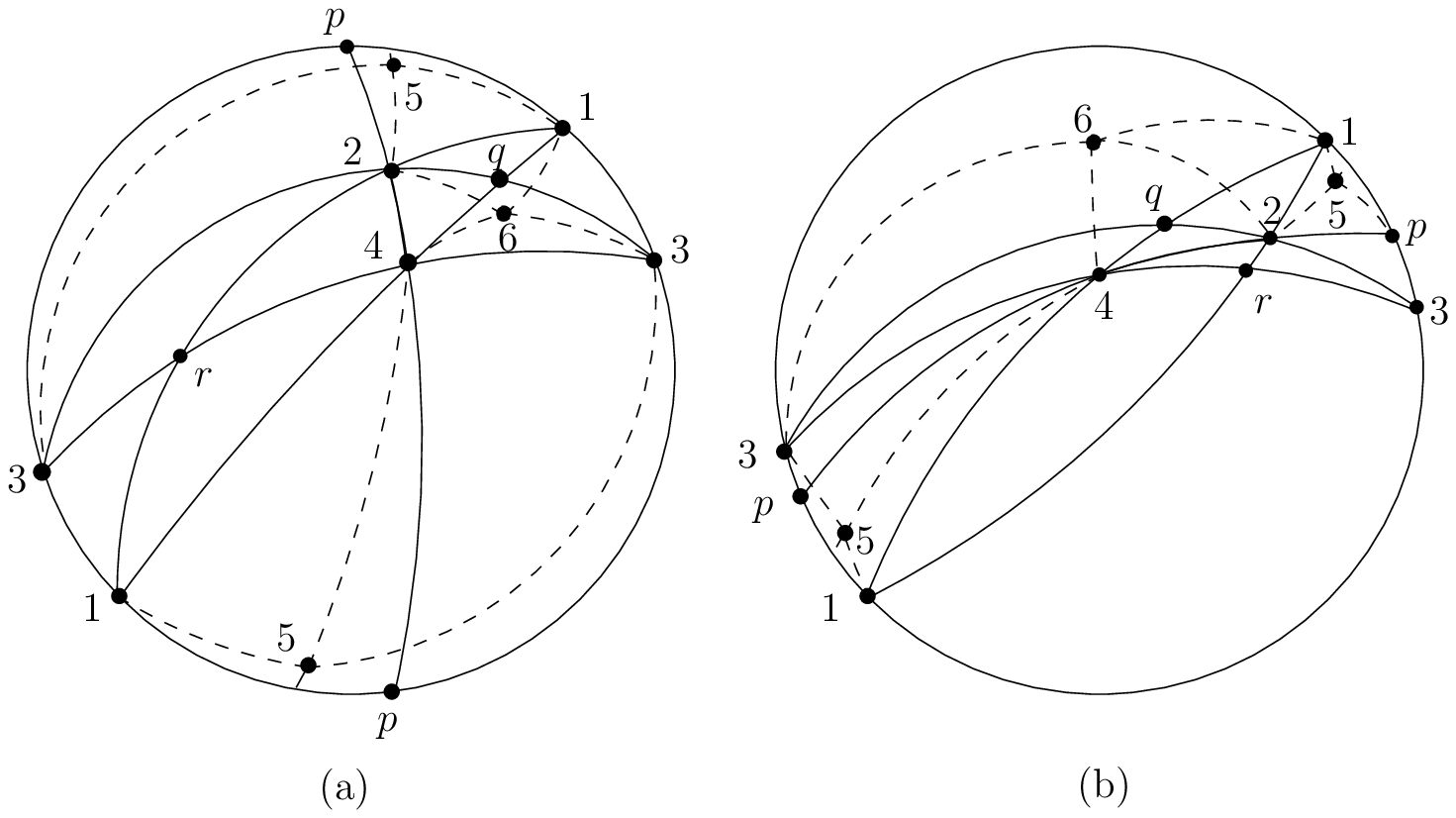}}\\
\scalebox{.6}{\includegraphics{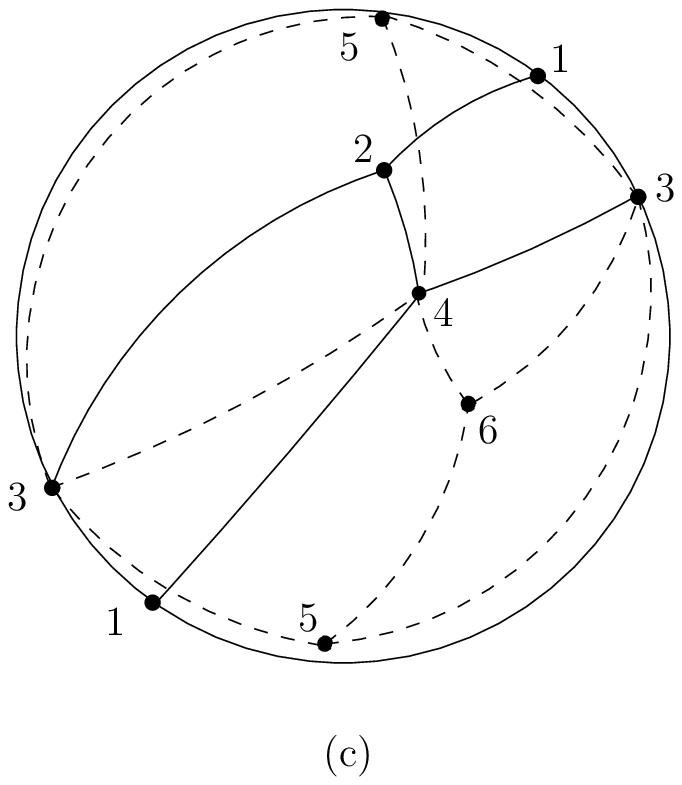}}
\end{center}
\caption{(a) and (b) Symmetric cases for constructing a projective
triangulation. (c) A canonical set always exists when six points in
$\mathcal{P}$ are in {\it general position}.\label{fig-init}}
\end{figure}

\begin{lemma}
If there exists a set of six points $($say, $\{1,2,3,4,5,6\})$ 
in a given point set $\mathcal{P}$ such that four 
of them $($say, $\mathcal{S}=\{1,2,3,4\})$ form a $K_4$-quadrangulation 
and the other two $($say, $\{5,6\})$ are in different regions of the projective triangulation 
formed by $\mathcal{S}$, then it is possible to triangulate the projective plane using 
these six points, unless $(n-2)$ points in $\mathcal{P}$ are collinear.  
\end{lemma} 
\begin{proof}
We give a constructive proof of the above statement. We first
construct a projective triangulation with the set $\{1,2,3,4\}$ (as
described above). Suppose points $5$ and $6$ lie in the regions
associated with the pseudo-points $p$ and $q$ respectively (see
Figure~\ref{fig-init}(a,b)). We now add point $5$ and make it
adjacent to the vertices of its bounding region, deleting the
pseudo-point $p$ and the edges it was incident with.  The newly added
edges are shown by dashed lines. The pseudo-points have also been kept
for better understanding.  We now add point $6$ and delete the
corresponding pseudo-point $q$ and the edges it was incident with. Now
we intend to delete the pseudo-point $r$ and construct a valid
projective triangulation using only points in $\mathcal{P}$. Here we
make the important observation that either the edge $12$ or $34$ can
be flipped. To see this, note that if flipping of neither of these
edges was possible, then $6$ must lie to the ``left" (as shown in
Figure~\ref{fig-init}(a)) or ``right" (as shown in
Figure~\ref{fig-init}(b)) of both the lines $52$ and $54$, in which
case flipping of edges would induce crossings. (Note that we refer to
a point being on the ``left" or ``right" of a line only {\it locally}
with respect to front half of the projective sphere.)  However, the
edge $24$ lies in between these two lines and $6$ cannot lie to its
left (resp. right). Thus, our claim holds.

Suppose the edge $12$ can be flipped. We then
construct a valid projective triangulation by flipping $12$, deleting
the pseudo-point $r$ and adding the edge $12$ in that region. Observe
that the projective triangulation constructed is isomorphic to that
shown in Figure~\ref{fig-irred}(b). In the event that flipping of
neither $12$ nor $34$ is possible, all four points $5, 2, 4, 6$ must
be collinear. Since such a flip is also not possible with any other
point in $\mathcal{P}$, they must all lie on the line $524$, implying
that $(n-2)$ of the points in $\mathcal{P}$ are collinear.
\end{proof}

We will refer to such a $K_4$-quadrangulation which has two points of $\mathcal{P}$ 
in different regions as a {\it canonical set}. 
We now have almost all the basic tools required for triangulating the real projective 
plane from a point set $\mathcal{P}$.
All that we need to characterize is the existence 
of a canonical set.  
So far we have not used anywhere the assumption that at least six points in $\mathcal{P}$ are 
in {\it general position}. It turns out that there always exists a 
canonical set in $\mathcal{P}$ in this case. We have the following lemma:

\begin{lemma}
\label{lemma4}
If at least six points in $\mathcal{P}$ are in general position, then there 
exists a canonical set.  
\end{lemma}  
\begin{proof}
We prove this lemma by the method of contradiction. We assume that the
lemma does not hold, so for every $K_4$-quadrangulation in
$\mathcal{P}$, all other points of $\mathcal{P}$ are in the same
region.  Consider a $K_4$-quadrangulation $\{1,2,3,4\}$ in
$\mathcal{P}$. Suppose we add two more points $5$ and $6$, and they
lie in the same region (as shown in Figure~\ref{fig-init}(c)). Now
consider the $K_4$-quadrangulation formed by $\{4,6,3,5\}$. If this is
to satisfy the property that all points in $\mathcal{P}$ lie in
exactly one of its regions, then it is easy to see that $2$ must lie
on or to the right of the line $54$. But now $2$ and $6$ lie in
different regions of the $K_4$-quadrangulation $\{4,5,1,3\}$, a
contradiction!
\end{proof}

So we now have a procedure for triangulating the real projective 
plane given a point set $\mathcal{P}$ with at least six points in general position. 
We summarize our results in the following theorem:

\begin{theorem}
Given a point set $\mathcal{P}=\{p_1,p_2,\ldots,p_n\}$ with 
at least six points in general position, it is 
always possible to construct a projective triangulation. 
\end{theorem}
We now present our algorithm which outputs a
triangulation of the real projective plane given a point set
$\mathcal{P} = \{p_1,p_2,\ldots,p_n\}$ with at least six points in
general position.

\begin{enumerate}
\item {} Find a set $\mathcal{S}=\{1,2,3,4,5,6\}$ of six points such that no 
       three points in $\mathcal{S}$ are collinear. 
\item {} Construct a projective triangulation with the set $\mathcal{S}$. Associate 
distinguishing planes with every triangle of the triangulation.
\item {} {\bf for} {all points $p \in \mathcal{P}\backslash \mathcal{S}$} {\bf do}
\item {} \hspace*{.5cm} Identify the triangle $\bigtriangleup abc$ in which $p$ lies. 
\item {} \hspace*{.5cm} Make $p$ adjacent to the vertices $a, b$ and $c$. 
	Make the distinguishing\\
 	\hspace*{.5cm} plane of $\bigtriangleup apb,
	\bigtriangleup bpc$, and $\bigtriangleup cpa$ the same as that
	for $\bigtriangleup abc$.
\item {} {\bf end for}
\item {} {\bf return}(triangulation of $\mathbb{P}^2$). 
\end{enumerate}

There are two possible approaches for finding the set $\mathcal{S}$ in
step 1. In the first approach, we arbitrarily choose a starting point
$q$ and initialize our set $\mathcal{S}=\{q\}$. For any point $p \in
\mathcal{P}\backslash\mathcal{S}$, we add $p$ in $\mathcal{S}$ if $p$
is not collinear with any two points in $\mathcal{S}$. We stop when
$\mathcal{S}$ contains six points. It may happen that we are not able
to find such a set $\mathcal{S}$ of six points if we start with any
random starting point $q$. So we iterate over all points in
$\mathcal{P}$ for choosing the starting point. This approach has a
worst-case time complexity of $\mathcal{O}(n^2)$. A slightly better
approach can be adopted for performing step 1, which works in
$\mathcal{O}(n)$ time if we assume that the minimum line cover of the
point set $\mathcal{P}$ is greater than $4$. In this approach, we
first choose any two points $1$ and $2$. Let the line defined by them
be $L$. We delete all other points in $\mathcal{P}$ on $L$. We now
choose two more points $3$ and $4$. Let the line defined by them be
$M$. We delete all other points in $\mathcal{P}$ on $M$. We also
delete all other points on $N$ (the line defined by $1$ and $3$) and
$T$ (the line defined by $2$ and $4$). Now choose two more points $5$
and $6$. We now have the required set $\mathcal{S} =
\{1,2,3,4,5,6\}$. It is easy to see that this approach takes
$\mathcal{O}(n)$ time if the minimum line cover of $\mathcal{P}$ is
greater than $4$. The above two approaches work reasonably well for
most point sets. However, for certain point sets, it may happen that
both these approaches fail to find such a set $\mathcal{S}$. We are
currently unaware of an optimal method for finding such a set which
works in all cases. We believe that some approach similar to that used
for solving the ``ordinary line" problem can be adopted for finding
the same (see for instance \cite{km-noldn-58,mah-olpcg-97,dm-focoh-99}). 
After having found such a set $\mathcal{S}$, we find a canonical set within 
$\mathcal{S}$ by a procedure similar to that described in Lemma \ref{lemma4}.

Once we have a canonical set, constructing the projective
triangulation in step 2 takes $\mathcal{O}(1)$ time. We store the
triangulation in a DCEL so that addition and deletion of edges and
vertices takes $\mathcal{O}(1)$ time. The loop in steps 3-6 runs once
for every point $p \in \mathcal{P}\backslash\mathcal{S}$.  Inside the
loop, we use our ``in-triangle" test (as described in
Section~\ref{sec-interior}) for testing whether a point lies inside a
given triangle. We use a procedure similar to that described by
Devillers et al. in \cite{dpt-wt-02} for identifying $\bigtriangleup
abc$ inside which $p$ lies. We first choose any arbitrary vertex $t$
of the current projective triangulation.  We then identify the
triangle whose interior is intersected by the line $L=tp$.  This test
is performed by checking for all edges $E$ of all triangles sharing
vertex $t$ whether the intersection of $L$ and the {\it line}
described by $E$ lies inside the given triangle (see
Figure~\ref{fig-walk}(a)). After having identified the starting
triangle, we move to its neighbor sharing the edge $E$. In this way,
we ``walk" in the triangulation along the line $L$.  We stop when $p$
lies inside the current triangle.  Although this method of ``walking"
in a triangulation has a worst-case time complexity of
$\mathcal{O}(n)$, it is reasonably fast for most practical purposes.
So the loop takes a total of $\mathcal{O}(n^2)$ steps. Thus, our
algorithm computes a projective triangulation from a given point set
$\mathcal{P}$ in $\mathcal{O}(n^2)$ steps.

\medskip

As mentioned in the introduction, the complexity of the algorithm is
not our main concern in the present paper. Still, note that our
algorithm is incremental, which is an important property in
practice. $\mathcal{O}(n^2)$ is a standard worst-case complexity for
incremental algorithms computing triangulations in the Euclidean
plane. After step~2, instead of inserting the points
incrementally, we could do the following\footnote{as suggested by an
anonymous reviewer}: for each point, find the triangular face of the
initial triangulation containing it. Then, in each of these faces,
triangulate the set of points using the usual affine method. This can
be done since the convex hulls of subsets of points in a triangular
face of the initial triangulation can be defined with the help of
distinguishing planes. This yields an optimal $O(n\log n)$ worst-case
time (non-incremental) algorithm.

\begin{figure}[htbp]
\centerline{
\scalebox{.6}{\includegraphics{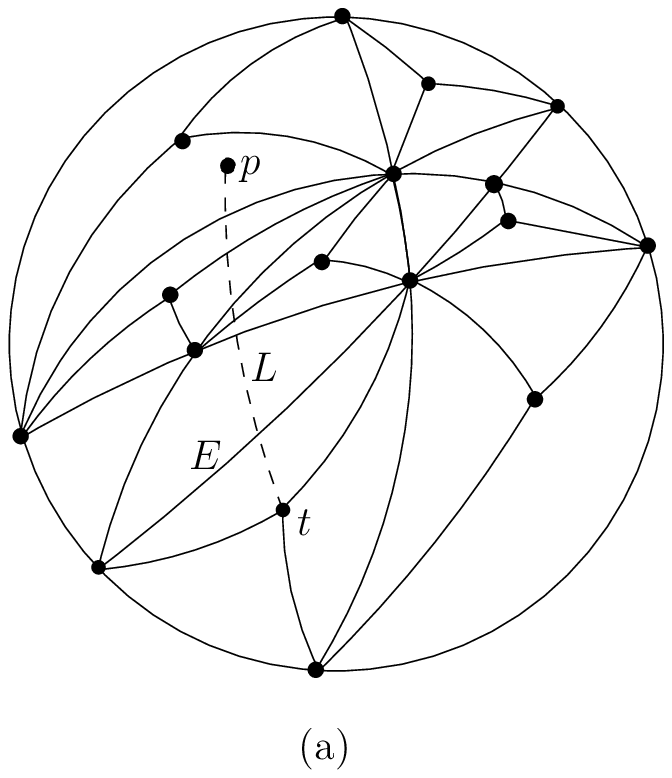}}
\scalebox{.6}{\includegraphics{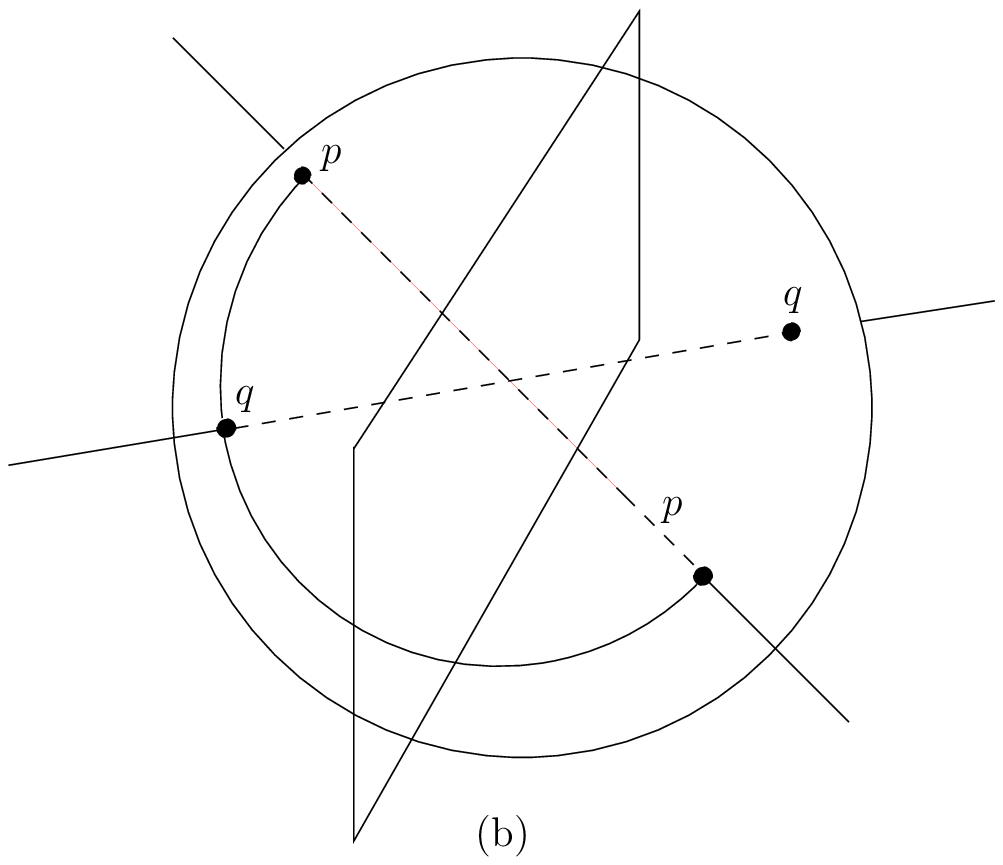}}
}
\caption{(a) ``Walking" in a projective triangulation. (b) Two copies
of the edge $pq$, only one is cut by the distinguishing plane.\label{fig-walk}}
\end{figure}

\section{Conclusion and Open Problems\label{sec-conclusion}}


It woud be interesting to check whether the metric on $\P^2$ allows
to define a triangulation of the projective plane that would extend
the notion of {\it Delaunay} triangulation, which is well-known in
the Euclidean setting. Then, extending the randomized incremental
insertion with a hierarchical data-structure such as \cite{d-dh-02} to
the projective case, if possible, would lead to an incremental
algorithm with better theoretical (and practical, too) complexity.

Also, problems like the {\it Minimum Weight Triangulation}
\cite{mr-mwtn-07}, {\it Minmax Length Triangulation}
\cite{et-qtaml-93}, etc., may have meaning even on the real projective
plane. The {\it Minimum Weight Triangulation} problem was neither
known to be NP-Hard nor solvable in polynomial time for a long time
\cite{gj-cigtn-79}. This open problem was recently solved and was
shown to be NP-Hard by Mulzer and Rote \cite{mr-mwtn-07}. The {\it
Minmax Length Triangulation} problem asks about minimizing the maximum
edge length in a triangulation of a point set $\mathcal{P}$. This
problem was shown to be solvable in time $\mathcal{O}(n^2)$ by
Edelsbrunner and Tan \cite{et-qtaml-93}. It would be interesting to
analyze the complexity of these problems on the real projective plane
$\mathbb{P}^2$.

\subsection*{Acknowledgements}

The authors would like to thank Olivier Devillers for his valuable
suggestions and helpful discussions.

\small

\end{document}